\PassOptionsToPackage{pdftex}{graphicx}
\documentclass{IEEEtran}
\usepackage{cite} 
\usepackage{array}
\usepackage{xcolor}
\usepackage{engord}
\usepackage{bm}
\usepackage{booktabs}
\usepackage{amssymb}
\usepackage{bbm}
\usepackage{blkarray}
\usepackage{graphicx}
\usepackage{balance}
\usepackage{multirow}
\usepackage{longtable}
\usepackage{rotating}
\usepackage{float}
\usepackage{arydshln}
\usepackage{amsthm}
\usepackage[caption=false,font=footnotesize]{subfig}
\usepackage[ruled]{algorithm2e} 
\usepackage[colorlinks=true, allcolors=blue]{hyperref}
\usepackage[most]{tcolorbox}
\usepackage{mathtools}

\newcommand{\mb}[1]{\mathbf{#1}}

\newcommand{\minimize}[1]{\underset{{#1}}{\text{minimize}}}

\newcommand{\st}{\text{subject to}}
\newcommand{\theoremsymbol}{\hfill \ensuremath{\triangleleft}}

\newtheoremstyle{nonboldnonitalic}  
  {}  
  {}  
  {\normalfont}  
  {}  
  {\itshape}  
  {.}  
  {.5em}  
  {}  

\theoremstyle{nonboldnonitalic}
\newtheorem{theorem}{Theorem}

\newtheorem{definition}[theorem]{Definition}

\newtheorem{assumption}{Assumption}
\theoremstyle{remark}

\ifCLASSINFOpdf
   \usepackage[pdftex]{graphicx}
\else
   \usepackage[dvips]{graphicx}
\fi

\makeatletter
\let\old@ps@headings\ps@headings
\let\old@ps@IEEEtitlepagestyle\ps@IEEEtitlepagestyle
\def\psccfooter#1{%
    \def\ps@headings{%
        \old@ps@headings%
        \def\@oddfoot{\strut\hfill#1\hfill\strut}%
        \def\@evenfoot{\strut\hfill#1\hfill\strut}%
    }%
    \def\ps@IEEEtitlepagestyle{%
        \old@ps@IEEEtitlepagestyle%
        \def\@oddfoot{\strut\hfill#1\hfill\strut}%
        \def\@evenfoot{\strut\hfill#1\hfill\strut}%
    }%
    \ps@headings%
}
\makeatother

\begin{document}
\begingroup
\allowdisplaybreaks

\title{
Differentially Private Obfuscation \\ of Power Grid Dynamics
}
\author{Shengyang Wu, \IEEEmembership{Student Member, IEEE}, and Vladimir Dvorkin, \IEEEmembership{Member, IEEE}
\thanks{Shengyang Wu ({\tt syseanwu@umich.edu}) and Vladimir Dvorkin ({\tt dvorkin@umich.edu}) are with the Department of Electrical Engineering and Computer Science, University of Michigan, MI 48109, USA. }}

\maketitle

\begin{abstract}
Dynamic models of power systems are critical for analyzing grid response to disturbances and blackouts, but the release of real-world dynamic models is hindered by privacy and cybersecurity concerns, as such models carry sensitive information about transmission, generation, and load parameters. We develop an algorithm for synthesizing dynamic grid models from real-world power grids balancing two objectives: the privacy of the source grid, quantitatively measured using the notion of differential privacy, and the fidelity of the synthesized model. The algorithm applies privacy-preserving noise to obfuscate the original grid parameters, but then optimizes the perturbed parameters to ensure that the resulting model dynamics are statistically consistent with those observed in the source grid. Application to the frequency dynamics of the IEEE 30-bus system reveals the inherent privacy-fidelity trade-off: stricter privacy requirements degrade modeling fidelity, yet optimization significantly improves the quality of the synthesized models.
\end{abstract}


\section{Introduction}
\IEEEPARstart{P}{ower} system dynamics governs how the grid responds to disturbances and rapidly changing conditions. With growing complexity from the integration of renewable energy resources \cite{Geng2022,Poolla2019} and AI electricity demand \cite{Li2024}, accurate dynamic models are imperative for resilient control \cite{Li2016, Mallada2017} and post-event analysis of blackouts \cite{ENTSOE_iberian_blackout_2025}. Yet, releasing real-world dynamic models is hindered by privacy and cybersecurity concerns. For example, disclosing generator models enables targeted AGC cyberattacks \cite{Tan2017}, while releasing frequency dynamics models exposes network topology and line parameters, thereby facilitating integrity attacks on power markets \cite{Xie2011}.

Kron reduction techniques \cite{Dorfler2012,Ishizaki2018,chevalier2022towards,mokhtari2025optimal} eliminate selected buses while preserving input-output dynamics, inherently obfuscating network topology and admittance data. However, the eliminated data remain vulnerable to reverse engineering under voltage and current measurements \cite{Low2024,Low2024_2}, while the reduced models can also seed data-driven topology estimation methods capable of recovering the original topology \cite{Deka2024}.

\textit{Differential privacy} (DP) \cite{Dwork2014} provides algorithmic foundations for privacy-preserving data publication. By perturbing the source data with carefully calibrated noise, DP renders any two adjacent datasets statistically indistinguishable in the random outcome of the perturbation, up to prescribed privacy parameters. DP has been leveraged for privacy-preserving releases of voltage and smart meter data streams from distribution systems, as well as records from PV and wind power plants \cite{campbell2026differentiallyprivatesyntheticvoltage,gough2021preserving,ravi2023solar,Dvorkin2023}. To protect grid data used in power system computations, DP has been applied to randomize power flow optimization solvers \cite{Dvorkin2020,Dvorkin2021}, dynamical simulations \cite{11462688}, training of machine learning models \cite{dvorkin2025privacy}, and aggregated market-clearing statistics \cite{Zhou2019}, i.e., applications where sensitive data may be inferred from computational outcomes. 

The key challenge of DP randomization is that the perturbed data may lose modeling fidelity, become infeasible for the underlying analysis, or introduce new system risks. Post-processing optimization addresses this by restoring modeling fidelity while preserving the privacy guarantee, and has been applied to DP release of demand and transmission data to ensure cost-consistent power flow outcomes \cite{Fioretto2020,Dvorkin2023} and to prevent the synthesized data from being used to calibrate cyberattacks \cite{Wu2025}. However, the application of post-processing optimization beyond steady-state power flow to power systems \textit{dynamics} remains largely unexplored. The control community has focused on the privacy-preserving release of measurements \cite{Katewa2020} and inputs \cite{Koufogiannis2017} for general dynamical systems, but no prior work has addressed DP release of the model itself. This gap is critical for power system dynamics, where model owners must simultaneously reproduce the electrical behavior of the source system and protect sensitive parameters (admittances, generation inertia and damping) that define it.

\textit{Contributions:} This work translates the success of DP in steady-state grid optimization \cite{Fioretto2020,Wu2025,Dvorkin2023} to power system dynamics. We develop algorithms for the public release of dynamic models with DP guarantees on sensitive parameters, such as inertia, damping, and admittance, while preserving modeling fidelity. Our technical contributions include:
\begin{enumerate}
    \item We develop a principled algorithm for the public release of power system dynamic models. In the context of grid frequency dynamics, we release the model $\dot{\omega}(t)=\tilde{f}(\omega(t),u(t))$ such that the original $f$ cannot be reverse-engineered from $\tilde{f}$, while the two models produce statistically similar frequency trajectories.
    
    \item To obfuscate grid topology, we introduce a DP version of Kron reduction that injects calibrated noise into the source admittance, making reverse engineering in style of \cite{Low2024,Deka2024} infeasible up to prescribed DP parameters.

    \item We formulate an ODE-constrained optimization to post-process DP grid parameters and recover modeling fidelity lost to DP noise. We propose a solution technique based on the adjoint method \cite{cao2003adjoint}, fitting the DP grid parameters to observable frequency trajectories, thus extending DP post-processing from steady-state grid optimization \cite{Fioretto2020,Wu2025,Dvorkin2023} to dynamical systems.
\end{enumerate}

Our algorithm contributes to the ongoing effort of synthesizing dynamic grid models, including the ENTSO-E model of the Pan-European grid \cite{Semerow2015}, where bus-level dynamics are initialized from standard models and tuned to match mean transients and dominant inter-area oscillation modes under a plant outage event. To support such efforts, we provide rigorous guarantees for source grids. Although focused on frequency dynamics as a case study, the approach generalizes to synchronous generator \cite{Karrari2004}, PSS \cite{Gurrala2010}, AVR \cite{Wu2020}, inverter models \cite{Geng2022}, and to linear systems more broadly.

\textit{Notation:} Lower- (upper-) case boldface letters denote column vectors (matrices); $\top$ stands for transposition;  $\mb{0}/\mb{1}$ are all-zero/one vectors. Scalar $a_{i}$ is the $i^{\text{th}}$ element of vector $\mb{a}$. The total and partial derivatives are denoted by $\mathrm{d}_x$ and $\partial_x$, respectively. Projection of $\mb{a}$ onto $[\underline{a},\overline{a}]$ is denoted $\mathcal{P}[\mb{a}]_{\underline{a}}^{\overline{a}}$.

\section{Preliminaries}

 In this section, we introduce the reference model, discuss the privacy risks of releasing its parameters, and overview the classic Kron-reduction technique for model reduction. 

\subsection{Transmission Frequency Dynamics}

Consider a power transmission network with a connected directed graph $(\mathcal{N}, \mathcal{E})$, with a set $\mathcal{N}$ of nodes (buses) and a set of $\mathcal{E}$ edges (transmission lines).  We define $(i,j)$ as the transmission line between node $i$ and $j$.

We make the following common assumptions \cite{chen2020distributed}:
\begin{enumerate}
    \item All lines $(i,j)\in \mathcal{E}$ are lossless and only characterized by their reactance $x_{ij}$. 
    \item Reactive power injections and flows are ignored.
    \item Voltage magnitude of all buses are fixed. 
\end{enumerate}

Based on these assumptions, this paper describes the linearized dynamics of the transmission network as
\begin{subequations}\label{eq:transmission dynamics}
    \begin{align}
    \mb{M}\dot{\bm{\omega}}&=
    \mb{K}
    \bm{\delta}+
        \mb{p}
    -\bm{\ell}-\mb{D}\bm{\omega}\\
    \dot{\bm{\delta}}&=\bm{\omega}\\
        \dot{\mb{p}}  &= -\mb{T}(\mb{p}+\mb{R}\bm{\omega})
    \end{align}
\end{subequations}
where the state variables include the vector $\bm{\omega}$ of nodal frequencies, the vector $\bm{\delta}$ of nodal voltage phase angles and the vector $\mb{p}$ of mechanical power output of every node, with entries corresponding to  non-generator nodes set to $0$. The disturbance is confined to load deviations, denoted by vector $\bm{\ell}$. All variables are deviations from their nominal values. 

The system design parameters include the diagonal matrices $\mb{M}$ of generator inertia, $\mb{D}$ of nodal damping, $\mb{T}$ with diagonal elements being the inverse of governor/turbine time constant, and $\mb{R}$ of droop control constants; for non-generator nodes, the corresponding entries in $\mb{M}, \mb{T}$ and $\mb{R}$ are zeros. The matrix $\mb{K}$ is a weighted Laplacian with elements
\begin{align}\label{eq:construct K}
    K_{ij}=
    \begin{cases}
        k_{ij}\coloneqq\frac{v_i^0 v_j^0}{x_{ij}} \cos(\delta_i^0 - \delta_j^0), & \text{if}\;(i,j)\in \mathcal{E},\\
        -\sum_{k} k_{ik},\;\forall (i,k)\in \mathcal{E}, & \text{if}\;i=j,\\
        0,  & \text{otherwise},
    \end{cases}
\end{align}
defining the sensitivity of power flows to voltage phase angle deviation from the nominal operating point.

\subsection{Risks of Releasing Model Parameters}

The frequency dynamics model \eqref{eq:transmission dynamics} includes five parameters: the turbine governor parameters $\mb{T}$ and $\mb{R}$, which are typically fixed and known, the generator inertia $\mb{M}$ and damping $\mb{D}$, which are proprietary to generation companies and system operators, and the weighted Laplacian $\mb{K}$, which encodes network topology and line admittances. Releasing the latter three carries privacy and security risks:

\subsubsection{Generator Inertia $\mb{M}$} In markets where inertia is explicitly procured, e.g., the UK stability market \cite{NESO_Inertia_Cost}, disclosing generator inertia exposes market bids for stability contracts and reveals unit commitment status, since committed units are the primary providers of inertia and spinning reserve. From a security perspective, \cite{acharya2020public} demonstrates that inertia data is instrumental for attackers to relocate system eigenvalues to destabilizing positions in the right-half plane.

\subsubsection{Damping Matrix $\mb{D}$} Damping reflects the aggregate effect of voltage control systems on generator response \cite{BergenVittal2000}, exposing proprietary parameters of voltage regulation.

\subsubsection{Weighted Laplacian $\mb{K}$} By construction, $\mb{K}$ encodes network topology and line admittances, and its release facilitates load redistribution attacks \cite{Liu2014,Liang2016}. Kron reduction offers a natural mechanism to obfuscate $\mb{K}$, which we overview next.

\subsection{Kron Reduction on Dynamical Models}

Applied to a network graph, Kron reduction eliminates selected nodes while preserving the electrical behavior at the remaining ones, forming a reduced network. The classic technique reduces the nodal current balance equations \cite{Dorfler2012}, while optimization-based variants can eliminate arbitrary buses while preserving desired network properties \cite{chevalier2022towards,mokhtari2025optimal}. We use Kron reduction to reduce the system in \eqref{eq:transmission dynamics}. Let $\gamma$ denote the set of boundary nodes to be preserved, and $\beta$ denote the interior nodes to be eliminated. We make further assumptions:
\begin{enumerate}
    \item All generator nodes are selected as boundary nodes.
    \item Interior node damping is negligible and can be ignored.
\end{enumerate}

Under these assumptions, equations \eqref{eq:transmission dynamics} are partitioned as
{\setlength{\arraycolsep}{4pt}
\begin{subequations}\label{eq:TSA compact}
    \begin{align}
    &\left[
    \begin{array}{c:c}
    \textbf{M}_{\gamma} & \mb{0} \\
    \hdashline
    \mb{0} & \mb{0}
    \end{array}
    \right]
    \begin{bmatrix}
        \bm{\omega}_\gamma \\
        \bm{\omega}_\beta
    \end{bmatrix}=
    \left[
    \begin{array}{c:c}
    \textbf{K}_{\gamma\gamma} & \textbf{K}_{\gamma\beta} \nonumber\\
    \hdashline
    \textbf{K}_{\beta\gamma} & \textbf{K}_{\beta\beta}
    \end{array}
    \right]
    \begin{bmatrix}
    \bm{\delta}_\gamma \\
    \bm{\delta}_\beta
    \end{bmatrix}+
    \begin{bmatrix}
        \mb{p}_\gamma \\
        \mb{0}
    \end{bmatrix}\\ 
    &\quad\quad\quad\quad\quad\quad\quad\quad\quad\quad\quad -\begin{bmatrix}
        \bm{\ell}_\gamma\\
        \bm{\ell}_\beta
    \end{bmatrix}
    -\left[
    \begin{array}{c:c}
    \textbf{D}_{\gamma} & \mb{0} \\
    \hdashline
    \mb{0} & \mb{0}
    \end{array}
    \right]
    \begin{bmatrix}
        \bm{\omega}_\gamma\\
        \mb{0}
    \end{bmatrix}\\[-0.7em]
    &\begin{bmatrix}
    \dot{\bm{\delta}}_\gamma \\
    \dot{\bm{\delta}}_\beta
    \end{bmatrix}
    = 
    \begin{bmatrix}
    \bm{\omega}_\gamma \\
    \bm{\omega}_\beta
    \end{bmatrix}\\
    &\begin{bmatrix}
        \dot{\mb{p}}_\gamma \\
        \mb{0}
    \end{bmatrix} = 
    \left[
    \begin{array}{c:c}
    \textbf{T}_{\gamma} & \mb{0} \\
    \hdashline
    \mb{0} & \textbf{T}_{\beta}
    \end{array}
    \right]
    \left(\begin{bmatrix}
        \mb{p}_\gamma \\
        \mb{0}
    \end{bmatrix}+
    \left[
    \begin{array}{c:c}
    \textbf{R}_{\gamma} & \mb{0} \\
    \hdashline
    \mb{0} & \textbf{R}_{\beta}
    \end{array}
    \right]
    \begin{bmatrix}
        \bm{\omega}_\gamma\\
        \bm{\omega}_\beta
    \end{bmatrix}
    \right)
    \end{align}
\end{subequations}
where the subscript $\cdot_\gamma$ denotes the variables and parameters corresponding to the boundary nodes, and $\cdot_\beta$ to interior nodes.} 

Kron reduction of \eqref{eq:TSA compact} eliminates the variables and parameters of interior nodes, giving rise to the following electrically equivalent reduced system:
\begin{subequations}\label{eq:ODE reduced}
    \begin{align}
    \mb{M}_\gamma\dot{\bm{\omega}}_\gamma&= \mb{K}_\text{red}\bm{\delta}_\gamma-\mb{D}_\gamma \bm{\omega}_\gamma+(\mb{p}-\bm{\ell}_\gamma)-\mb{K}_\text{ac}\bm{\ell}_\beta \label{eq:ODE reduced_one}\\
    \dot{\bm{\delta}}_\gamma &= \bm{\omega}_\gamma\\
    \dot{\mb{p}}_\gamma
    &= -\mb{T}_\gamma(
        \mb{p}_\gamma
    +\mb{R}\bm{\omega}_\gamma)
\end{align}
\end{subequations}
where $\mb{K}_\text{red}=\mb{K}_{\gamma\gamma}-\mb{K}_{\gamma \beta}\mb{K}_{\beta\beta}^{-1}\mb{K}_{\beta\gamma}$ is the reduced coefficient matrix, and $\mb{K}_\text{ac}=-\mb{K}_{\gamma\beta}\mb{K}_{\beta\beta}^{-1}$ is the accompanying matrix that maps the load disturbance of the reduced nodes onto the boundary nodes. The reduced dynamic model, characterized by $\mb{K}_\text{red}$ and $\mb{K}_\text{ac}$, presents several appealing characteristics that facilitate the release of dynamical models:  1) For a reduced model that embeds parameters of a reduced network, multiple original networks can yield the same reduction, inherently providing a naive obfuscation of topology; 2) The reduced model \eqref{eq:ODE reduced} exhibits the same behavior as the original model \eqref{eq:TSA compact}, as Kron reduction is a linear, lossless operation \cite{Dorfler2012}.

\subsection{Privacy Risks in Kron Reduction -- A Numerical Example}\label{sec:reverse KR}
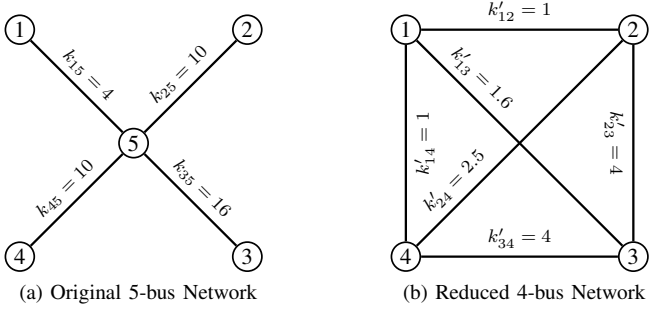
\begin{figure}
    \centering
    \subfloat[Original 5-bus Network]{%
    \resizebox{0.2\textwidth}{!}{
    \begin{tikzpicture}
    \node [draw=black, line width = 0.75pt, circle, inner sep = 1.5] (5) at (0,0) {5};
    \node [draw=black, line width = 0.75pt, circle, inner sep = 1.5] (1) at (-1.75,1.75) {1};
    \node [draw=black, line width = 0.75pt, circle, inner sep = 1.5] (2) at (1.75,1.75) {2};
    \node [draw=black, line width = 0.75pt, circle, inner sep = 1.5] (4) at (-1.75,-1.75) {4};
    \node [draw=black, line width = 0.75pt, circle, inner sep = 1.5] (3) at (1.75,-1.75) {3};

    \draw [black,line width = 1pt] (1) -- node[pos=0.5,above,sloped,font=\footnotesize] {$k_{15}=4$} (5);
    \draw [black,line width = 1pt] (2) -- node[pos=0.5,above,sloped,font=\footnotesize] {$k_{25}=10$} (5);
    \draw [black,line width = 1pt] (3) -- node[pos=0.5,above,sloped,font=\footnotesize] {$k_{35}=16$} (5);
    \draw [black,line width = 1pt] (4) -- node[pos=0.5,above,sloped,font=\footnotesize] {$k_{45}=10$} (5);
\end{tikzpicture}
}
\label{fig:example KR og}
}
\hfill
    \subfloat[Reduced 4-bus Network]{%
    \resizebox{0.2\textwidth}{!}{
    \begin{tikzpicture}\label{fig:example KR 2}
    \node [draw=black, circle, line width = 0.75pt, inner sep = 1.5] (1) at (-1.75,1.75) {1};
    \node [draw=black, circle, line width = 0.75pt, inner sep = 1.5] (2) at (1.75,1.75) {2};
    \node [draw=black, circle, line width = 0.75pt, inner sep = 1.5] (4) at (-1.75,-1.75) {4};
    \node [draw=black, circle, line width = 0.75pt, inner sep = 1.5] (3) at (1.75,-1.75) {3};

    \draw [black,line width = 1pt] (1) -- node[pos=0.5,above,sloped,font=\footnotesize] {$k_{12}'=1$} (2);
    \draw [black,line width = 1pt] (2) -- node[pos=0.5,below,sloped,font=\footnotesize] {$k_{23}'=4$} (3);
    \draw [black,line width = 1pt] (3) -- node[pos=0.5,above,sloped,font=\footnotesize] {$k_{34}'=4$} (4);
    \draw [black,line width = 1pt] (4) -- node[pos=0.5,below,sloped,font=\footnotesize] {$k_{14}'=1$} (1);
    \draw [black,line width = 1pt] (1) -- node[pos=0.25,above,sloped,font=\footnotesize] {$k_{13}'=1.6$} (3);
    \draw [black,line width = 1pt] (2) -- node[pos=0.75,above,sloped,font=\footnotesize] {$k_{24}'=2.5$} (4);
\end{tikzpicture}
}
        }
    \caption{Kron reduction of a 5-bus network. Bus 5 is the interior node.}
    \label{fig:example KR}
\end{figure}

Figure \ref{fig:example KR} illustrates Kron reduction of a 5-bus network via elimination of one interior node (bus 5), yielding a reduced 4-bus network with equivalent dynamics but a modified nodal and topological structure. Despite this obfuscation, the original network parameters can still be reverse-engineered under certain conditions. Following \cite{Low2024,Low2024_2}, suppose the original topology is known. Our goal is to recover the original parameters $k$ of the source network from the edge parameters $k'$ of the reduced network. The weighted Laplacians of the original and reduced networks are given by
{\setlength{\arraycolsep}{1pt}
\begin{align}
    \mb{K}&=\left[
    \begin{array}{c:c}
    \textbf{K}_{\gamma\gamma} & \textbf{K}_{\gamma\beta} \\
    \hdashline
    \textbf{K}_{\beta\gamma} & \textbf{K}_{\beta\beta}
    \end{array}
    \right]=\left[
\begin{array}{cccc:c}
 -k_{15} & 0       & 0       & 0       & k_{15} \\
 0       & -k_{25} & 0       & 0       & k_{25} \\
 0       & 0       & -k_{35} & 0       & k_{35} \\
 0       & 0       & 0       & -k_{45} & k_{45} \\
 \hdashline
 k_{15}  & k_{25}  & k_{35}  & k_{45}  & S
\end{array}
\right], \nonumber
\end{align}
where $S=-(k_{15}+k_{25}+k_{35}+k_{45})$, and
\begin{align}
\mb{K}_\text{red}&=\left[
\begin{array}{cccc}
 -\sum_{j\neq1} k'_{1j} & k'_{12} & k'_{13} & k'_{14} \\
 k'_{21}          & -\sum_{j\neq 2} k'_{2j} & k'_{23} & k'_{24} \\
 k'_{31}          & k'_{32} & -\sum_{j\neq 3} k'_{3j} & k'_{34} \\
 k'_{41}          & k'_{42}       & k'_{43}   & -\sum_{j\neq 4} k'_{4j}
\end{array}
\right],\nonumber
\end{align}
respectively. 
}
From Kron reduction we know the relationship between the original and Kron-reduced parameters, i.e.,  
$\mb{K}_\text{red}=\mb{K}_{\gamma\gamma}-\mb{K}_{\gamma \beta}S^{-1}\mb{K}_{\gamma \beta}^\top$, giving rise to the following system of equations:
\begin{subequations}\label{eq:example Fx}
\begin{align}
k'_{12}\,(k_{15}+k_{25}+k_{35}+k_{45}) - k_{15}k_{25} = 0 \\
k'_{13}\,(k_{15}+k_{25}+k_{35}+k_{45}) - k_{15}k_{35} = 0 \\
k'_{14}\,(k_{15}+k_{25}+k_{35}+k_{45}) - k_{15}k_{45} = 0 \\
k'_{23}\,(k_{15}+k_{25}+k_{35}+k_{45}) - k_{25}k_{35} = 0 \\
k'_{24}\,(k_{15}+k_{25}+k_{35}+k_{45}) - k_{25}k_{45} = 0 \\
k'_{34}\,(k_{15}+k_{25}+k_{35}+k_{45}) - k_{35}k_{45} = 0
\end{align}
\end{subequations}
which is an overdetermined system with less variables than equations. It has an analytic solution recovering $\mb{k}$ from $\mb{k}_\text{red}$: 
\begin{subequations}
    \begin{align}\label{eq:example sol}
        k_{5m}&=y_m(y_1+y_2+y_3+y_4),\\
        &\text{with} \; y_i=\sqrt{\frac{k'_{ij} k'_{in}}{k_{jn}'}},\quad \forall i\neq j\neq n.
    \end{align}
\end{subequations}
This can be validated by taking the values of $\mb{k}_\text{red}$ in Fig.\ref{fig:example KR 2} into \eqref{eq:example sol}, leading to 
\begin{align*}
    (y_1,y_2,y_3,y_4)&=(\frac{\sqrt{10}}{5},\frac{\sqrt{10}}{2},\frac{4\sqrt{10}}{5},\frac{\sqrt{10}}{2})\\
    (k_{15},k_{25},k_{35},k_{45})&=(4,10,16,10),
\end{align*}
which matches the original parameters in Fig.\ref{fig:example KR og}. This substantiates the risks of recovering the original network data from a Kron-reduced network. The success of analytic recovery depends on the structure of the original and reduced networks: 
\begin{enumerate}
    \item More edges in the reduced network (this example) yield an overdetermined system analogous to \eqref{eq:example Fx}, which can be solved either analytically or via the Gauss-Newton method \cite{Gratton2007} by minimizing the residual.
    \item With the same number of edges in the reduced and original networks, the system of equations is well-determined, and the original network can be recovered analytically. This is the case of the classic Y-$\Delta$ transformation, where the Y-connection is uniquely determined from the $\Delta$-connection.
    \item With fewer edges in the reduced network, the system is underdetermined with infinitely many solutions, and additional information such as measurements is needed to identify a physically meaningful solution.
\end{enumerate}

\section{Differentially Private Release\\ of Frequency Dynamics Model}

To control privacy risks in releasing grid dynamics models, we leverage DP. In this section, we overview the DP fundamentals \cite{Dwork2014} and present the proposed algorithm for the release of grid frequency dynamics model.

\subsection{Background on Differential Privacy}

Given a vector $\mb{x}\in\mathbb{R}^{n}$ stacking the parameters of model \eqref{eq:transmission dynamics},  the DP goal is to produce its synthetic counterpart $\tilde{\mb{x}}$ that yields similar frequency dynamics yet does not disclose the original parameters up to prescribed privacy guarantees. Formally, releasing $\tilde{\mb{x}}$ must make $\mb{x}$ indistinguishable from an adjacent vector $\mb{x}'$ in the following sense. 

\begin{definition}[Adjacency]\label{def:adjacency} Two vectors $\mb{x},\mb{x}'\in\mathbb{R}^{n}$ are $\alpha-$adjacent, for some $\alpha>0$, if $\exists i\in\{1,\dots,n$\}, such that $x_{j}=x_{j}',\forall j\in\{1,\dots,n\}\backslash i$, and $|x_{i}-x_{i}'|\leqslant\alpha$. That is, they are different in one item by at most $\alpha$. \theoremsymbol
\end{definition}

DP guarantees that an adversary cannot distinguish between $\mb{x}$ and $\mb{x}'$ when observing $\tilde{\mb{x}}$, even when it knows all but one element of $\mb{x}$, up to a privacy loss parameter $\varepsilon>0$. The vector-Laplace mechanism is commonly used to establish this guarantee \cite{Dvorkin2023}. Let $\text{Lap}(s)^{n}$ denote an $n-$dimensional zero-mean Laplace random variable with scale parameter $s$. Then, the DP synthetic release takes the form $\tilde{\mb{x}}=\mb{x}+\text{Lap}(\frac{\alpha}{\varepsilon})^{n}$, i.e., the original vector is perturbed by a calibrated Laplace noise. This satisfies the following definition of $\varepsilon-$DP.

\begin{definition}[$\varepsilon-$DP]\label{def:DP} 
The Laplace mechanism is called $\varepsilon-$DP if, for any outcome $\tilde{\mb{x}}$ in the range of the mechanism, and any two $\alpha-$adjacent vectors $\mb{x}$ and $\mb{x}'$, the ratio of probabilities of observing the same outcome is bounded as
\begin{align}
    \left|\log\left(\frac{\text{Pr}[\mb{x}\textcolor{white}{'} + \text{Lap}(\alpha/\varepsilon)^{n}=\tilde{\mb{x}}]}{\text{Pr}[\mb{x}' + \text{Lap}(\alpha/\varepsilon)^{n}=\tilde{\mb{x}}]}\right)\right|\leqslant\varepsilon,
\end{align}
i.e.,  the log-likelihood ratio associated with any outcome $\tilde{\mb{x}}$ under two $\alpha$-adjacent vectors is uniformly bounded by $\varepsilon$.
\theoremsymbol
\end{definition}
A smaller privacy loss $\varepsilon$ makes the outcomes on adjacent vectors more indistinguishable, thus stronger protecting their difference. Stronger guarantees, however, come at the cost of larger noise, as the Laplace scale $\alpha/\varepsilon$ grows as $\varepsilon$ decreases. 

We next present two important properties of DP.  

\begin{theorem}[Sequential composition \cite{Dwork2014}]\label{th:SC} A series of $k$ $\varepsilon_{i}-$DP releases of $\mb{x}$ satisfies $\sum_{i=1}^{k}\varepsilon_{i}-$DP.
\end{theorem}

\begin{theorem}[Post-processing immunity \cite{Dwork2014}]\label{th:PP} If $\tilde{\mb{x}}$ satisfies $\varepsilon$-DP, then $g\circ\tilde{\mb{x}}$, where $g$ is any data-independent post-processing of the query answer, also satisfies $\varepsilon$-DP. 
\end{theorem}

The former bounds the privacy loss across multiple computations on the original data, and the latter states that any data-independent transformation of the DP release answer preserves the privacy guarantee.  The two are important building blocks.

\subsection{DP Release of Frequency Dynamics Models}\label{subsec:dp_release}

We seek to synthesize a DP counterpart of \eqref{eq:ODE reduced_one}, of the form:
\begin{align}
    \tilde{\mb{M}}_\gamma^0\dot{\bm{\omega}}_\gamma&= \tilde{\mb{K}}_\text{red}^0\bm{\delta}_\gamma-\tilde{\mb{D}}_\gamma^0 \bm{\omega}_\gamma+(\mb{p}-\bm{\ell}_\gamma)-\tilde{\mb{K}}_\text{ac}^0\bm{\ell}_\beta,\label{eq:dp_system}
\end{align}
such that the synthetic inertia, damping, and Kron-reduced Laplacian matrices are indistinguishable from their source counterparts, while producing statistically similar frequency trajectories. Towards the goal, we apply the vector-Laplace mechanism from above. Let vectors $\mb{m}=\{M_i,\forall i\in \gamma\}$ and $\mb{d}=\{D_i,\forall i\in \gamma\}$ collect the diagonal elements of $\mb{M}$ and $\mb{D}$, respectively, and let $\mb{k}$ collect the edge weights forming $\mb{K}$. The DP obfuscation proceeds in the following two steps:
\begin{subequations}\label{eq:Lap_perturbation}
\begin{enumerate}
    \item \textit{Inertia and damping obfuscation:} for some adjacency $\alpha_{m},\alpha_{d}>0$, apply the vector-Laplace mechanism:
        \begin{align}
            \tilde{\mb{m}}^0&=\mb{m} + \text{Lap}\left(\frac{\alpha_m}{\varepsilon/3}\right)^{\!\!|\gamma|\!\!}\!\!, \; \tilde{\mb{d}}^0=\mb{d} + \text{Lap}\left(\frac{\alpha_d}{\varepsilon/3}\right)^{\!\!|\gamma|}\!\!\!\!\!\!\label{eq:lap input pertubation}
        \end{align}
    to get perturbed $\tilde{\mb{M}}_\gamma^0=\text{diag}[\tilde{\mb{m}}^0]$ and $\tilde{\mb{D}}_\gamma^0=\text{diag}[\tilde{\mb{d}}^0].$
    \item \textit{Network obfuscation:} for some adjacency $\alpha_{k}>0$, apply the vector-Laplace mechanism to original weights:
    \begin{align}
        \tilde{\mb{k}}^0&=\mathcal{P}_{\mathbb{R}_{>0}^{|\mathcal{E}|}}\left[\mb{k} + \text{Lap}\left(\frac{\alpha_k}{\varepsilon/3}\right)^{|\mathcal{E}|}\right]\label{eq:step2}
    \end{align}
    to get the perturbed weighted Laplacian $\tilde{\mb{K}}^0$ as in \eqref{eq:construct K}. Then, perform Kron reduction of $\tilde{\mb{K}}^0$ w.r.t. selected boundary and interior nodes and get $\tilde{\mb{K}}_\text{red}^0$ and $\tilde{\mb{K}}_\text{ac}^0.$
\end{enumerate}
\end{subequations}

The second step is the DP Kron reduction. Since the perturbation is applied directly to the edge weights prior to forming the Laplacian, the resulting $\tilde{\mb{K}}$ preserves the symmetry condition required for Kron reduction. Since the Laplace noise is unbounded, the perturbed weights $\tilde{\mb{k}}$ may take negative values, violating the physical requirement of positive line admittances and breaking the nonsingularity required for Kron reduction. To resolve this issue, the perturbed weights are projected onto strictly positive. 

The release of the perturbed model \eqref{eq:dp_system} is accompanied by an $\varepsilon$-DP guarantee per the following result. 

\begin{theorem}[$\varepsilon$-DP of frequency dynamics model]
    Steps 1--2 produce an $\varepsilon$-DP release of the frequency dynamics model for $\alpha_k$-adjacent power flow sensitivity, $\alpha_m$-adjacent inertia, and $\alpha_d$-adjacent damping vectors. \theoremsymbol \label{th:DP} 
\end{theorem}

\begin{proof}
    Per the Laplace mechanism (Definition \ref{def:DP}) and sequential composition (Theorem \ref{th:SC}), the computation in \eqref{eq:lap input pertubation} is $\tfrac{2}{3}\varepsilon$-DP. Forming the diagonal matrices in Step 1 preserves this guarantee by post-processing immunity (Theorem \ref{th:PP}). Similarly, the $\tfrac{1}{3}\varepsilon$-DP guarantee of Step 2 is preserved under the projection in \eqref{eq:step2} and subsequent Kron reduction, both of which are data-independent linear post-processing operations. Sequential composition of the two steps accumulates a total privacy loss of $\left(\tfrac{1}{3}+\tfrac{2}{3}\right)\varepsilon$, yielding the $\varepsilon$-DP guarantee. 
\end{proof}

\subsection{Post-Processing Optimization for Modeling Fidelity}

Theorem \ref{th:DP} guarantees privacy but not fidelity: the dynamics of \eqref{eq:dp_system} may deviate significantly from those of \eqref{eq:ODE reduced_one}. To ensure modeling fidelity, we propose to post-process synthetic inertia, damping and network parameters while preserving the privacy guarantee of Theorem \ref{th:DP}. That is, we seek the smallest adjustment of $\tilde{\mb{m}}^0, \tilde{\mb{d}}^0$ and $\tilde{\mb{k}}^0$ that replicates the frequency dynamics of the source model. We make the following assumption.
\begin{assumption}[Public knowledge]
    The frequency trajectories $\hat{\bm{\omega}}(0),\dots, \hat{\bm{\omega}}(T)$ of the source model are publicly observable.
\end{assumption}
This assumption is reasonable, as frequency trajectories under severe transients are already publicly released, e.g., after the 2025 Iberian Peninsula blackout \cite{ENTSOE_iberian_blackout_2025}. By post-processing immunity (Theorem \ref{th:PP}), fitting the DP model to such public trajectories does not affect the DP guarantee of Theorem \ref{th:DP}.

The proposed post-processing optimization takes the form:
\begin{subequations}\label{eq:pp full}
    \begin{align}
    \minimize{\tilde{\mb{k}},\tilde{\mb{m}}, \tilde{\mb{d}},\bm{\omega},\bm{\delta},\mb{p}}&   \int_0^T \!\!\left\|\bm{\omega}(t)-\hat{\bm{\omega}}(t)\right\|_2^2 \mathrm{d}t 
    + \rho
    \left\|
    \begin{bmatrix*}
        \tilde{\mb{k}} - \tilde{\mb{k}}_\text{r}^0\\
        \tilde{\mb{m}} - \tilde{\mb{m}}^0\\
        \tilde{\mb{d}} - \tilde{\mb{d}}^0
    \end{bmatrix*}
    \right\|_{2}^2 \label{eq:pp full obj}\\[-0.3em] 
    \st& \nonumber\\
    \text{diag}[\tilde{\mb{m}}]&\dot{\bm{\omega}}(t)= \mb{C}_{\text{r}}\text{diag}[\tilde{\mb{k}}]\mb{C}_{\text{r}}^{\top}\bm{\delta}(t) - \text{diag}[\tilde{\mb{d}}]\bm{\omega}(t) \nonumber\\
    &\textcolor{white}{\dot{\bm{\omega}}(t)=}\quad\quad\quad\quad\;\;+\mb{p}(t)-\bm{\ell}_\gamma - \tilde{\mb{K}}_\text{ac}^0\bm{\ell}_\beta, \label{eq:pp full con 1}\\
    &\dot{\bm{\delta}}(t) = \bm{\omega}(t), \;\! \dot{\mb{p}} (t) = \!-\mb{T}(\mb{p}(t) +\mb{R}\bm{\omega}(t)),\\
    &\bm{\omega}(0) = 0, \; \bm{\delta}(0) = 0, \; \mb{p}(0) = 0.  \label{eq:pp full con 3}
    \end{align}
\end{subequations}
which jointly optimizes the model parameters $\tilde{\mb{k}},\tilde{\mb{m}}, \tilde{\mb{d}}$ and the frequency states $\bm{\omega},\bm{\delta},\mb{p}$ coupled through the ODE constraints \eqref{eq:pp full con 1}--\eqref{eq:pp full con 3}. The objective minimizes the distance between the public $\hat{\bm{\omega}}$ and resulting $\bm{\omega}$ frequency trajectories over a period $T$, improving modeling fidelity, with a regularization term weighted by $\rho>0$ anchoring the solution to the Laplace perturbation \eqref{eq:Lap_perturbation}. Constraint \eqref{eq:pp full con 1} replicates the swing equation \eqref{eq:dp_system}, where $\tilde{\mb{m}}$ and $\tilde{\mb{d}}$ are diagonalized to form the inertia and damping matrices, and the Kron-reduced Laplacian $\tilde{\mb{K}}_{\text{red}}=\mb{C}_{\text{r}}\text{diag}[\tilde{\mb{k}}]\mb{C}_{\text{r}}^{\top}$ is reconstructed from the optimized edge weights $\tilde{\mb{k}}$ and the node-edge incidence matrix $\mb{C}_{\text{r}}$ of the reduced network. For simplicity, the accompanying matrix $\tilde{\mb{K}}_\text{ac}$ is fixed to its DP value $\tilde{\mb{K}}_\text{ac}^0$ from Step 2 in \eqref{eq:Lap_perturbation} and is not post-processed; accordingly, $\tilde{\mb{k}}$ in \eqref{eq:pp full} refers exclusively to the edge weights of $\tilde{\mb{K}}_{\text{red}}$, with $\tilde{\mb{k}}_\text{r}^0$ from $\tilde{\mb{K}}_{\text{red}}^0=\mb{C}_{\text{r}}\text{diag}[\tilde{\mb{k}}_{\text{r}}^{0}]\mb{C}_{\text{r}}^{\top}$ serving as the regularization reference. 

The ODE-constrained problem \eqref{eq:pp full} is not amenable to off-the-shelf solvers. We propose a gradient-based solution via the adjoint method \cite{cao2003adjoint}. Let $\bm{\theta}=(\tilde{\mb{k}},\tilde{\mb{m}}, \tilde{\mb{d}})$ collect the model parameters and $\mb{x}=(\bm{\omega},\bm{\delta},\mb{p})$ the system states. Problem \eqref{eq:pp full} is written compactly as:
\begin{align}
    \minimize{\bm{\theta},\mb{x}}\;& J(\mb{x},\bm{\theta}) \coloneqq \int_0^T \!\!\!\!\left\|\mb{H}\mb{x}(t)\!-\hat{\bm{\omega}}(t)\right\|_2^2\mathrm{d}t  + \rho\|\bm{\theta}-\tilde{\bm{\theta}}^0\|_2^2 \nonumber\\
    \st \quad&  \dot{\mb{x}}(t)= f(\mb{x}(t),\bm{\theta}),  \quad \mb{x}(0)=\mb{0}, \label{eq:pp compact}
\end{align}
where $\mb{H}$ extracts the frequency states from $\mb{x}$. The proposed algorithm initializes $\bm{\theta}$ from the Laplace perturbation outcome $\tilde{\bm{\theta}}^0$ in \eqref{eq:Lap_perturbation} and updates it via gradient descent:
\begin{align}
    \bm{\theta} \leftarrow \bm{\theta} - \eta\, \mathrm{d}_{\bm{\theta}} J(\mb{x}, \bm{\theta}),\label{eq:theta_update}
\end{align}
where $\eta > 0$ is the step size  (learning rate). By the chain rule, the total gradient $\mathrm{d}_{\bm{\theta}} J$ decomposes as:
\begin{align}
    \mathrm{d}_{\bm{\theta}} J = \int_0^T \!\!\left(\partial_{\mathbf{x}} l\, \mathrm{d}_{\bm{\theta}} \mathbf{x} + \partial_{\bm{\theta}} l\right) \mathrm{d}t 
    \label{eq:grad_J}
\end{align}
where $l(\mathbf{x},\bm{\theta}) = \|\mathbf{H}\mathbf{x}(t) - \hat{\bm{\omega}}(t)\|_2^2 + \frac{\rho}{T}\|\bm{\theta} - \tilde{\bm{\theta}}^0\|_{2}^{2}$. The term $\partial_{\bm{\theta}} l$ is straightforward to compute. The difficulty lies in $\mathrm{d}_{\bm{\theta}} \mathbf{x}$, which captures how the states respond to parameter changes through the ODE constraint in \eqref{eq:pp compact}. Differentiating \eqref{eq:pp compact} with respect to $\bm{\theta}$ yields the sensitivity ODE:
\begin{align}
    \mathrm{d}_t \mathrm{d}_{\bm{\theta}} \mathbf{x} = \partial_{\mathbf{x}} f\, \mathrm{d}_{\bm{\theta}} \mathbf{x} + \partial_{\bm{\theta}} f, \quad \mathrm{d}_{\bm{\theta}}\mathbf{x}(0) = \mathbf{0}, \label{eq:sens_ODE}
\end{align}
which must be integrated forward in time for each component of $\bm{\theta}$. Since $|\bm{\theta}|$ grows with network size, this is computationally intractable at scale. The adjoint method resolves this by replacing the $|\bm{\theta}|$ sensitivity ODEs with a single adjoint ODE of the same dimension as $\mathbf{x}$, from which the integral term in \eqref{eq:grad_J} is recovered without computing $\mathrm{d}_{\bm{\theta}} \mathbf{x}$ explicitly.

To derive the total derivative via the adjoint method, we formulation the Lagrangian function of \eqref{eq:pp compact}:
\begin{align}
    \mathcal{L} \coloneqq \int_{0}^{T} \big[ l(\mb{x},\bm{\theta}) - \bm{\lambda}^{\top} \left(\dot{\mb{x}}-f(\mb{x},\bm{\theta})\right) \big] \mathrm{d}t \;- \bm{\mu}^{\top}\mb{x}(0),\label{eq:lagrangian}
\end{align}
where $\bm{\lambda}$ is the time-dependent vector of Lagrange multipliers (also termed \textit{adjoint variables}) associated with the ODE, and $\bm{\mu}$ is the time-invariant multiplier associated with the initial condition. Since constraints of \eqref{eq:pp compact} are equalities,  we have $\mathrm{d}_{\bm{\theta}} \mathcal{L}=\mathrm{d}_{\bm{\theta}} J$, so the total derivative is expressed from \eqref{eq:lagrangian} as
\begin{align}
    \mathrm{d}_{\bm{\theta}} J&=\int_{0}^{T}\!\!\! \left[ 
    \partial_{\mb{x}} l \,\mathrm{d}_{\bm{\theta}}\mb{x} + \partial_{\bm{\theta}}l -\bm{\lambda}^\top(\mathrm{d}_{\bm{\theta}}\dot{\mb{x}}- \partial_{\mb{x}}f \,\mathrm{d}_{\bm{\theta}}\mb{x} - \partial_{\bm{\theta}}f)   
    \right] \mathrm{d}t \nonumber \\ 
    &\quad- \bm{\mu}^\top \mathrm{d}_{\bm{\theta}}\mb{x}(0) \nonumber\\ 
    &=\int_{0}^{T}\!\!\! \left[ 
    \partial_{\mb{x}} l \mathrm{d}_{\bm{\theta}}\mb{x} + \partial_{\bm{\theta}}l + \bm{\lambda}^\top (\partial_{\mb{x}}f \mathrm{d}_{\bm{\theta}}\mb{x} + \partial_{\bm{\theta}}f) \right] \mathrm{d}t \nonumber \\ 
    &\quad- \bm{\mu}^\top \mathrm{d}_{\bm{\theta}}\mb{x}(0) - \bm{\lambda}^\top  \mathrm{d}_{\bm{\theta}}\mb{x} (t) \Big|_{0}^{T} + \int_{0}^{T}\!\!\! \big[ \dot{\bm{\lambda}}^{\top} \mathrm{d}_{\theta}\mb{x}\big]\mathrm{d}t \nonumber\\
    &=\int_{0}^{T}\!\!\!\! \big[ ( \partial_{\mb{x}} l \!+\! \bm{\lambda}^\top\partial_{\mb{x}}f \!+\! \dot{\bm{\lambda}})^{\top}\mathrm{d}_{\bm{\theta}}\mb{x}\big] \mathrm{dt} \!+\!\! \int_{0}^{T}\!\!\!\! \big[\partial_{\bm{\theta}}l \!+\!\bm{\lambda}^\top\partial_{\bm{\theta}}f \big] \mathrm{d}t \nonumber\\
    &\quad- \bm{\lambda}(T)^{\top} \mathrm{d}_{\bm{\theta}}\mb{x}(T) + (\bm{\lambda}(0)-\bm{\mu})^\top \mathrm{d}_{\bm{\theta}}\mb{x}(0),\label{eq:adjoint total derivative}
\end{align}
which still contains the gradient $\mathrm{d}_{\bm{\theta}} \mathbf{x}$ difficult to compute. We eliminate it using the following adjoint equations \cite{cao2003adjoint}:
\begin{subequations}\label{eq:adjoint equations}
    \begin{align}
        \partial_{\mb{x}} l + \bm{\lambda}^\top\partial_{\mb{x}}f + \dot{\bm{\lambda}}=\mb{0},\label{eq:adjoint_ode}\\
        \bm{\lambda}(T)=\mb{0},\quad \bm{\lambda}(0)=\bm{\mu},
    \end{align}
\end{subequations}
where \eqref{eq:adjoint_ode} is an ODE for the adjoint variable $\boldsymbol{\lambda}(t)$, 
integrated backward in time from $t=T$ to $t=0$ with terminal condition 
$\boldsymbol{\lambda}(T)=\mathbf{0}$, and $\boldsymbol{\mu} = \boldsymbol{\lambda}(0)$ 
is the multiplier on the initial state constraint. Plugging \eqref{eq:adjoint equations} into  \eqref{eq:adjoint total derivative} simplifies the gradient as 
\begin{align}
    \mathrm{d}_{\bm{\theta}} J = \int_{0}^{T} \big[\partial_{\bm{\theta}}l + \bm{\lambda}^\top\partial_{\bm{\theta}}f \big] \mathrm{d}t,\label{eq:gradient}
\end{align}
where $\partial_{\bm{\theta}}l$ and $\partial_{\bm{\theta}}f$ are straightforward partial derivatives of known functions, and $\bm{\lambda}(t)$ is obtained by solving the adjoint ODE \eqref{eq:adjoint_ode}  backward in time. The total cost of computing the gradient is one forward ODE solve for $\mb{x}(t)$ and one backward ODE solve for $\bm{\lambda}(t)$, regardless of the dimension of $\bm{\theta}$. 

\subsection{The Algorithm}
\begin{algorithm}[t!]
\small
\KwIn{Source model: $\mb{k},\mb{m},\mb{d};$ graph data: $\beta,\gamma;$ public data $\hat{\bm{\omega}};$ DP parameters: $\alpha, \varepsilon;$ optimization data: $\rho, \eta, I.$
  }

  \vspace{0.25em}
  {\footnotesize \textbf{1}} Initial Laplace obfuscation of the source model
  \begin{subequations}
  \begin{align}
      &\tilde{\mb{k}}^0=\mathcal{P}\big[\mb{k} + \text{Lap}\left(\alpha_k/(\varepsilon/3)\right)^{|\mathcal{E}|}\big]_{\underline{\mb{k}}}^{\overline{\mb{k}}}, \\
      &\tilde{\mb{m}}^0\!\!=\mathcal{P}\big[\mb{m} + \text{Lap}\left(\alpha_m/(\varepsilon/3)\right)^{|\gamma|}\big]_{\underline{\mb{m}}}^{\overline{\mb{m}}}, \\
      &\tilde{\mb{d}}^0=\mathcal{P}\big[\mb{d} + \text{Lap}\left(\alpha_d/(\varepsilon/3)\right)^{|\gamma|}\big]_{\underline{\mb{d}}}^{\overline{\mb{d}}}
  \end{align}
  \end{subequations}

  {\footnotesize \textbf{2}} Kron reduction of obfuscated model $\tilde{\mb{K}} = \mb{C}\,\text{diag}[\tilde{\mb{k}}^0]\,\mb{C}^{\top}$
    \begin{subequations}
        \begin{align}
            \tilde{\mb{K}}_\text{red}^0& \leftarrow \tilde{\mb{K}}_{\gamma\gamma}-\tilde{\mb{K}}_{\gamma \beta}\tilde{\mb{K}}_{\beta\beta}^{-1}\tilde{\mb{K}}_{\beta\gamma}, \quad \tilde{\mb{K}}_\text{ac} \leftarrow -\tilde{\mb{K}}_{\gamma\beta}\tilde{\mb{K}}_{\beta\beta}^{-1}
        \end{align}
    \end{subequations}

  {\footnotesize \textbf{3}} Post-processing optimization
    \For{$i = 1,\dots,I$}{
        \vspace{0.25em}
        Compute the gradient $\mathrm{d}_{\bm{\theta}} J^i$ of \eqref{eq:pp full} using \eqref{eq:gradient} 

        \vspace{0.25em}
        Take the projected gradient step 
        \begin{subequations}\label{eq:grad g}
            \begin{align}
                \tilde{\mb{k}}_\text{r}^{i} &\leftarrow \mathcal{P}\big[\tilde{\mb{k}}_\text{r}^{i-1} - \eta \cdot \mathrm{d}_{\mb{k}_\text{r}} J^i\big]_{\underline{\mb{k}}}^{\overline{\mb{k}}}\\
                \tilde{\mb{m}}^{i} &\leftarrow \mathcal{P}\big[\tilde{\mb{m}}^{i-1} - \eta \cdot \mathrm{d}_{\mb{m}} J^i\big]_{\underline{\mb{m}}}^{\overline{\mb{m}}}\\
                \tilde{\mb{d}}^{i} &\leftarrow \mathcal{P}\big[\tilde{\mb{d}}^{i-1} - \eta \cdot \mathrm{d}_{\mb{d}} J^i\big]_{\underline{\mb{d}}}^{\overline{\mb{d}}}
            \end{align}
        \end{subequations}
        \vspace{-0.5cm}
}
\KwOut{Laplacian $\tilde{\mb{K}}_{\text{red}}=\mb{C}_\text{r} \operatorname{diag}(\tilde{\mb{k}}_\text{r}^{I})\mb{C}_\text{r}^{\top}$, acc. matrix $\tilde{\mb{K}}_{\text{ac}}$, inertia $\tilde{\mb{M}}_{\gamma}=\text{diag}[\tilde{\mb{m}}^{I}]$, damping $\tilde{\mb{D}}_{\gamma}=\text{diag}[\tilde{\mb{d}}^{I}]$}
  \caption{DP Release of Grid Dynamics Model}\label{alg:pp-GD}
\end{algorithm}

The algorithm for DP release of the frequency dynamics model is summarized in Alg.~\ref{alg:pp-GD}. It takes the parameters $(\mb{k},\mb{m},\mb{d})$ of the source model, the sets of boundary and interior nodes $\gamma$ and $\beta$, the public frequency trajectories $\hat{\bm{\omega}}$, all from the source grid. It further requires the adjacency parameters $\alpha=(\alpha_k,\alpha_m,\alpha_d)$ and the privacy loss $\varepsilon$. Finally, it takes the regularization parameter $\rho$, the gradient descent step size $\eta$ and the maximum number of iterations $I$.

Step 1 applies the vector-Laplace mechanism for initial obfuscation, projecting the perturbed parameters onto physically valid ranges. Step 2 performs Kron reduction of the obfuscated network. Step 3 post-processes the obfuscated parameters via gradient descent to recover modeling fidelity, projecting iterates onto the same valid ranges. Upon termination, the algorithm outputs DP parameters of the reduced system \eqref{eq:dp_system}.

The algorithm preserves the privacy guarantee of Theorem~\ref{th:DP}. The DP guarantee is established at Step~1, where projection onto physically valid parameter ranges preserves it by post-processing immunity (Theorem~\ref{th:PP}). Kron reduction and projected gradient descent in Steps~2 and~3 operate on public data only and therefore do not diminish the guarantee.

\section{Numerical Experiments}

\subsection{Settings}

\textit{System:} We use the IEEE 30-bus system with 6 generators and 21 loads with a total load of $283.4$~MW at the nominal operating point. The dynamic model parameters are set uniformly: generator inertia $M=10$~p.u.\,power$\cdot$s/p.u.\,frequency, damping $D=5$~p.u.\,power/p.u.\,frequency, governor time constant $T=20$~s$^{-1}$, and droop $R=0.04$~p.u. For topology obfuscation, the original 30-bus network with 37 transmission lines (Fig. \ref{fig:IEEE 30 system og}) is Kron-reduced treating generator buses as boundary nodes.  The reduced network comprises of 6 buses and 15 synthetic transmission lines (Fig. \ref{fig:IEEE 30 system red}). To simulate dynamics, we use Euler's method with a time step of $0.01$~s.

\textit{DP parameters:} The adjacency is set to $\alpha_k=\alpha_m=\alpha_d=1$~p.u. for power flow sensitivities, inertia, and damping, respectively. Three privacy loss levels $\varepsilon\in\{0.5,1,2\}$ represent strong, medium, and weak privacy guarantees, respectively.

\textit{Post-processing:} A single public frequency trajectory $\hat{\bm{\omega}}$ is available for post-processing, corresponding to a uniform $20\%$ load increase at all nodes. The first $30$~s of this trajectory are used in the post-processing optimization \eqref{eq:pp full}, while the fidelity of the model is tested on 10 random events. In Alg. \ref{alg:pp-GD}, the range of model parameters is fixed to $\mb{k}\in [1,100]$, $\mb{m}\in [1,40]$ and $\mb{d}\in [1,40]$; the iteration limit $I=500$ and step size $\eta=100.$ The large step size compensates for the order-of-magnitude disparity between frequency and parameter values.

\begin{figure}
    \centering
    \includegraphics[width=0.8\linewidth]{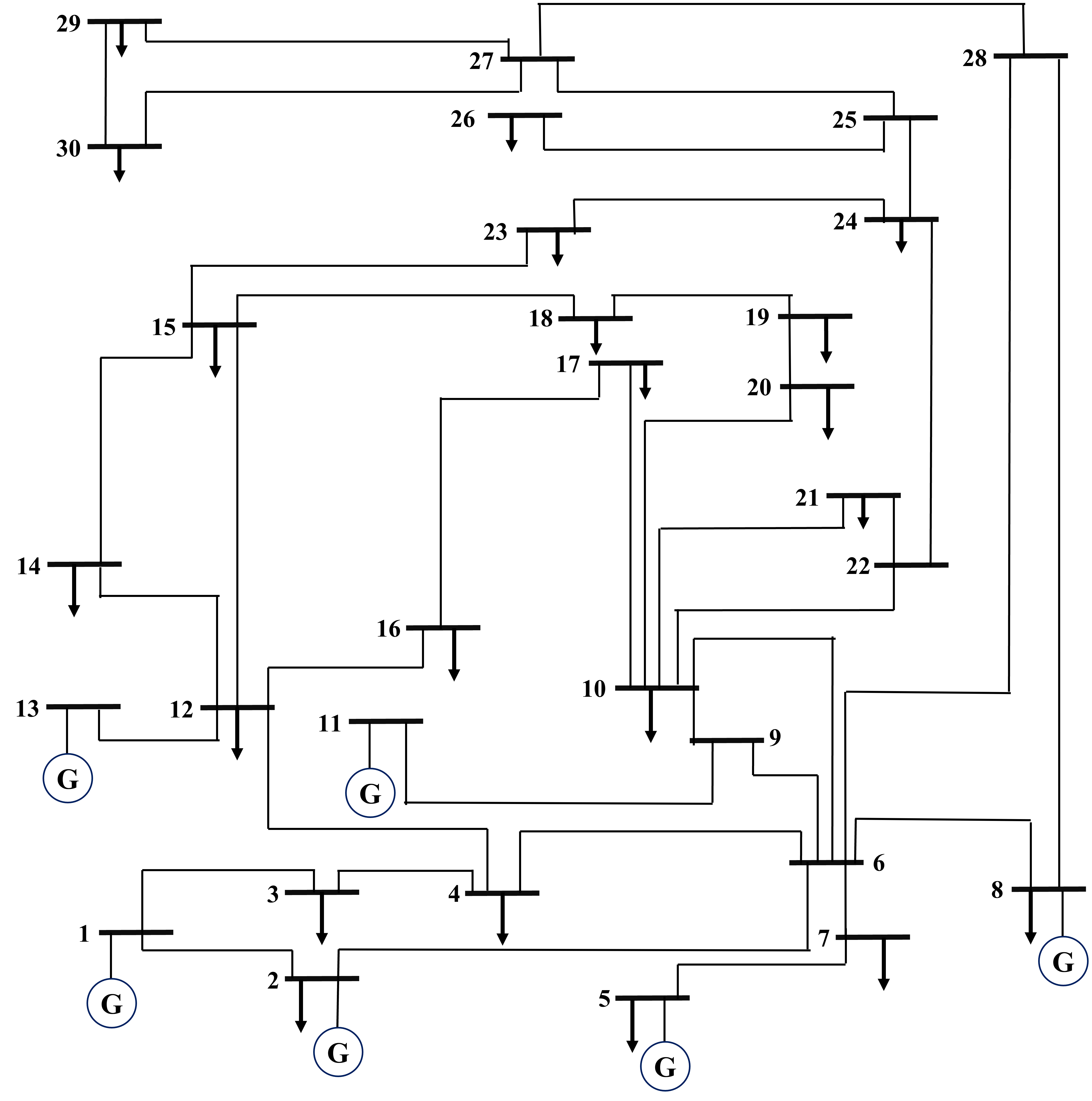}
    \caption{The full IEEE 30-bus system}
    \label{fig:IEEE 30 system og}
\end{figure}

 \begin{figure}
    \centering
    \includegraphics[width=0.8\linewidth]{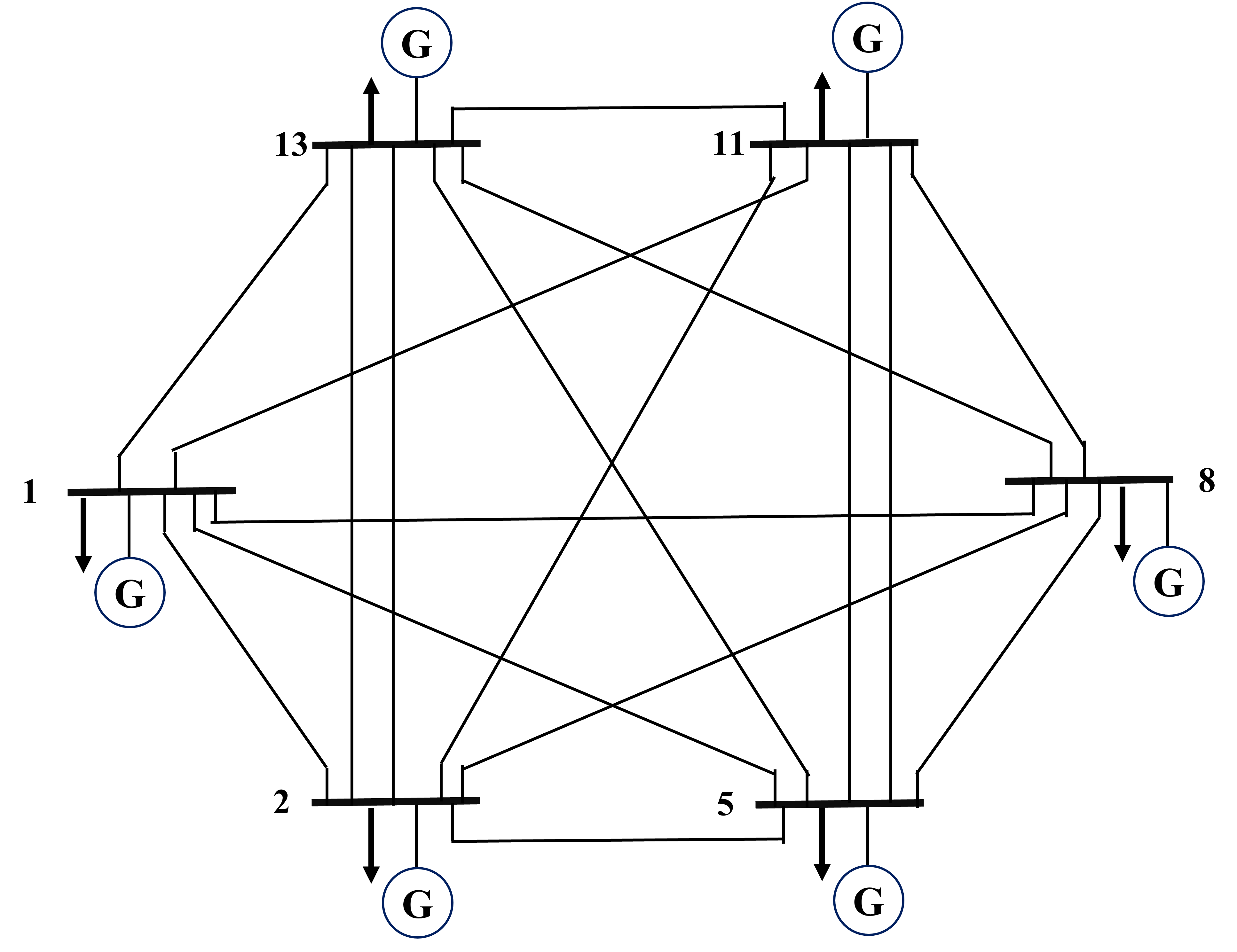}
    \caption{Kron-reduced IEEE 30-bus system}
    \label{fig:IEEE 30 system red}
\end{figure}

\subsection{Evaluation}
For each privacy level $\varepsilon$, Alg.~\ref{alg:pp-GD} is run $100$ times, yielding $100$ synthetic frequency dynamics models. The frequency trajectory mismatch $\int_0^T\|\mathbf{H}\mathbf{x}(t)-\hat{\boldsymbol{\omega}}(t)\|_2^2\,\mathrm{d}t$ over gradient descent iterations is shown in Fig.~\ref{fig:loss_distribution}. Higher initial mismatch appears at stronger privacy requirements (smaller $\varepsilon$), as Step~1 injects more noise into the source parameters. In all cases, however, the mismatch is reduced to nearly zero within $100$ iterations, confirming fidelity recovery across privacy levels.

\begin{figure}
    \centering
    \includegraphics[width=0.95\linewidth]{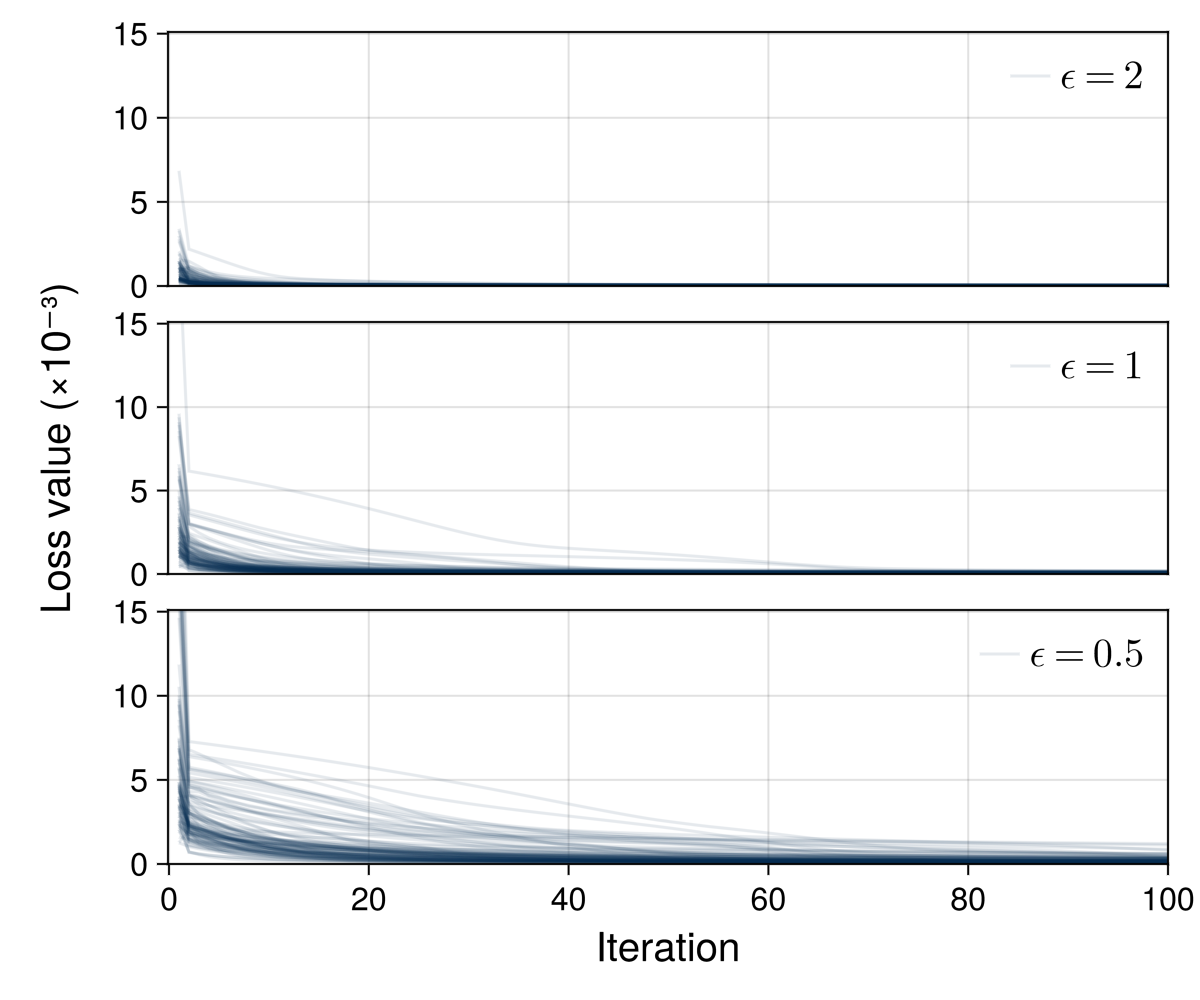}
    \caption{Trajectory mismatch $\int_0^T\|\mathbf{H}\mathbf{x}(t)-\hat{\boldsymbol{\omega}}(t)\|_2^2\,\mathrm{d}t$ over gradient descent iterations for $100$ runs of Alg.~\ref{alg:pp-GD} under three privacy levels $\varepsilon\in\{0.5,1,2\}$.}
    \label{fig:loss_distribution}
\end{figure}
 
The fidelity of the synthesized models is measured by the average trajectory mismatch over $10$ random transients:
\begin{subequations}
    \begin{align}
        L_{\text{lap}}&=\frac{1}{10}\sum_{s=1}^{10}\int_{0}^{30}\|\boldsymbol{\omega}_{\text{lap}}^{s}(t)-\boldsymbol{\omega}_{\text{ref}}^{s}(t)\|_2^2 \; \mathrm{d}t, \\
        L_{\text{pp}}&=\frac{1}{10}\sum_{s=1}^{10}\int_{0}^{30}\|\boldsymbol{\omega}_{\text{pp}}^{s}(t)-\boldsymbol{\omega}_{\text{ref}}^{s}(t)\|_2^2 \; \mathrm{d}t,
    \end{align}
\end{subequations}
where $s$ indexes the transient, $\boldsymbol{\omega}_{\text{ref}}^{s}$ is the true trajectory, $\boldsymbol{\omega}_{\text{lap}}^{s}$ is produced by Steps~1--2 of Algorithm~\ref{alg:pp-GD} (Laplace perturbation and Kron reduction only), and $\boldsymbol{\omega}_{\text{pp}}^{s}$ includes Step~3 (post-processing optimization). The results are shown in Fig. \ref{fig:loss_test_transients}. Without post-processing, $L_{\text{lap}}$ averages $1.5\times10^{-4}$ at $\varepsilon=0.5$, corresponding to roughly $0.012$~Hz of pointwise frequency deviation. Post-processing reduces this to $L_{\text{pp}}\approx2\times10^{-5}$, or about $0.004$~Hz, i.e., a threefold improvement in fidelity. As $\varepsilon$ increases (privacy requirement relaxes), both $L_{\text{lap}}$ and $L_{\text{pp}}$ further decrease, confirming the privacy-fidelity trade-off.

\begin{figure}
    \centering
    \includegraphics[width=0.99\linewidth]{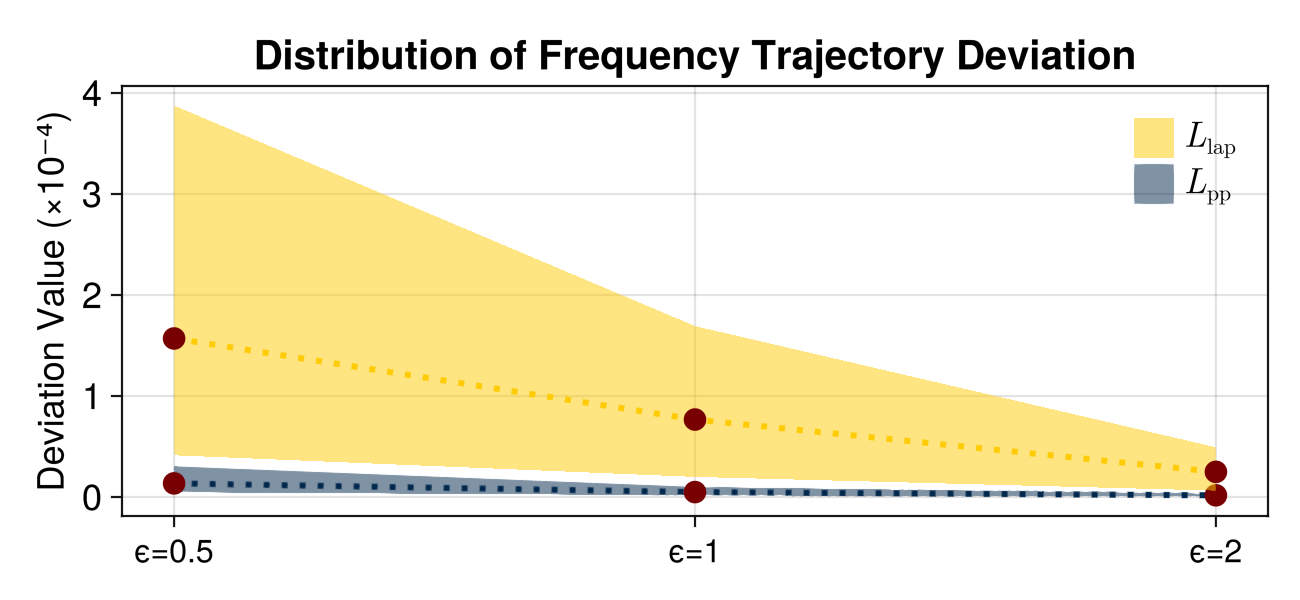}
    \caption{Distribution of $L_{\text{lap}}$ (yellow) and $L_{\text{pp}}$ (blue) across $100$ synthetic models and $10$ randomly generated transients, under three privacy levels $\varepsilon\in\{0.5,1,2\}$. Shaded areas show the $80\%$ intervals; dotted lines show the mean.}
    \label{fig:loss_test_transients}
\end{figure}

Figure~\ref{fig:freq_distribution} shows the frequency trajectories of $100$ synthetic models at $\varepsilon=0.5$ across five out-of-sample transients: Laplace perturbation alone (yellow), post-processed (blue), and the true model (red). While both share the same privacy guarantee, Laplace perturbation alone yields trajectories that deviate substantially from the true one. Post-processing restores fidelity: the blue trajectories are tightly concentrated around the true trajectory, with negligible mean error. Fidelity is highest in transient~1, whose load pattern most closely resembles the post-processing reference $\hat{\bm{\omega}}$. Performance degrades modestly under dissimilar transients (notably nodes~11 and~13 in transient~5), yet the post-processed models consistently outperform the Laplace-perturbed ones across all transients.
 
\begin{figure*}[t]
    \centering
    \includegraphics[width=0.99\linewidth]{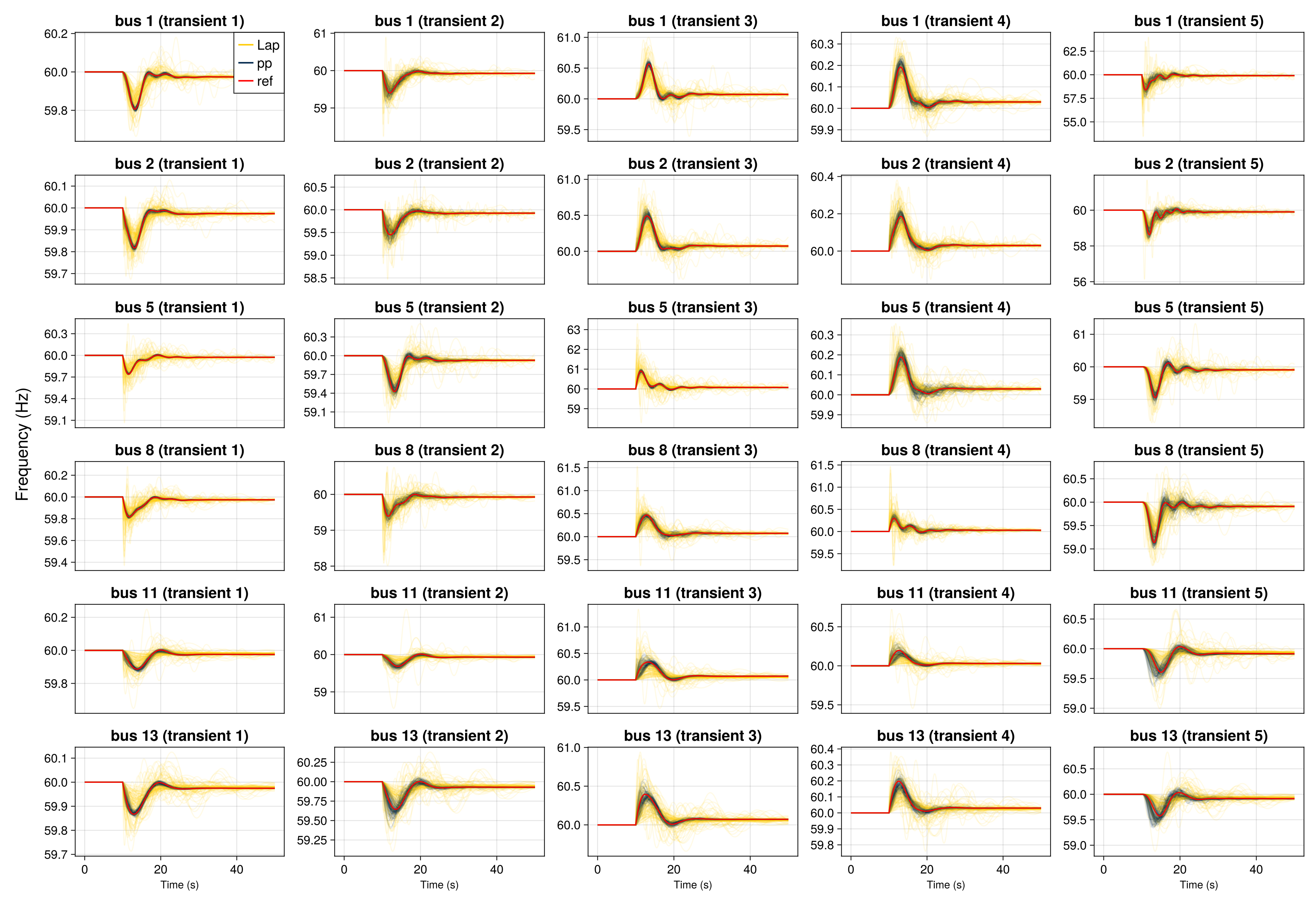}
    \caption{Frequency trajectories of $100$ synthetic models at $\varepsilon=0.5$ across five out-of-sample transients, all starting at $t=10$~s: transient~1, uniform $5\%$ load increase at all nodes ($0.15$~p.u. total); transient~2, load increase of $0.2$~p.u. at nodes~3 and~4; transient~3, load increase of $0.2$~p.u. at nodes~5 and~20; transient~4, full load loss at nodes~10 and~30 ($0.17$~p.u. total); transient~5, load addition of $0.5$~p.u. at bus~1. Yellow: Laplace-perturbed model (Steps~1--2 only). Blue: post-processed model (Steps~1--3). Red: true model.}
    \label{fig:freq_distribution}
\end{figure*}


Figures~\ref{fig:k_distribution}--\ref{fig:D_distribution} show the distributions of the synthetic model parameters $\tilde{\mb{k}}_{\text{red}}$, $\tilde{\mb{M}}_\gamma$, and $\tilde{\mb{D}}_\gamma$ across $100$ samples, before and after post-processing, under each privacy level $\varepsilon$. Fig.~\ref{fig:k_distribution} shows the distribution of the power flow sensitivity vector $\mb{k}$ under different privacy levels $\varepsilon$. As $\varepsilon$ decreases (i.e., stronger privacy protection), more noise is injected, leading to a wider spread in the distributions. This effect is reflected in the larger interquartile ranges and longer whiskers of both the blue and yellow box plots. For certain lines, such as 2--3 and 5--6, the post-processing optimization substantially shifts the values of $\mb{k}$, indicating that the power flow sensitivities on these lines are essential for restoring the true frequency trajectories. Some power flow sensitivities like $k_{56}$ and $k_{46}$ are especially important, as the interquartile range of the post-processed sensitivities is even wider than that before post-processing. In contrast, for lines such as 1--2, the distributions before and after post-processing remain similar, suggesting that the changes of $\mb{k}$ on these lines are not necessary to recover the system dynamics. Finally, both the Laplace-perturbed sensitivities $\tilde{\mb{k}}_{\text{red}}^{0}$ and the post-processed sensitivities $\tilde{\mb{k}}_{\text{red}}$ have mean values close to the true sensitivities, which is consistent with the zero-mean property of the Laplace noise.

 \begin{figure}
     \centering
     \includegraphics[width=0.99\linewidth]{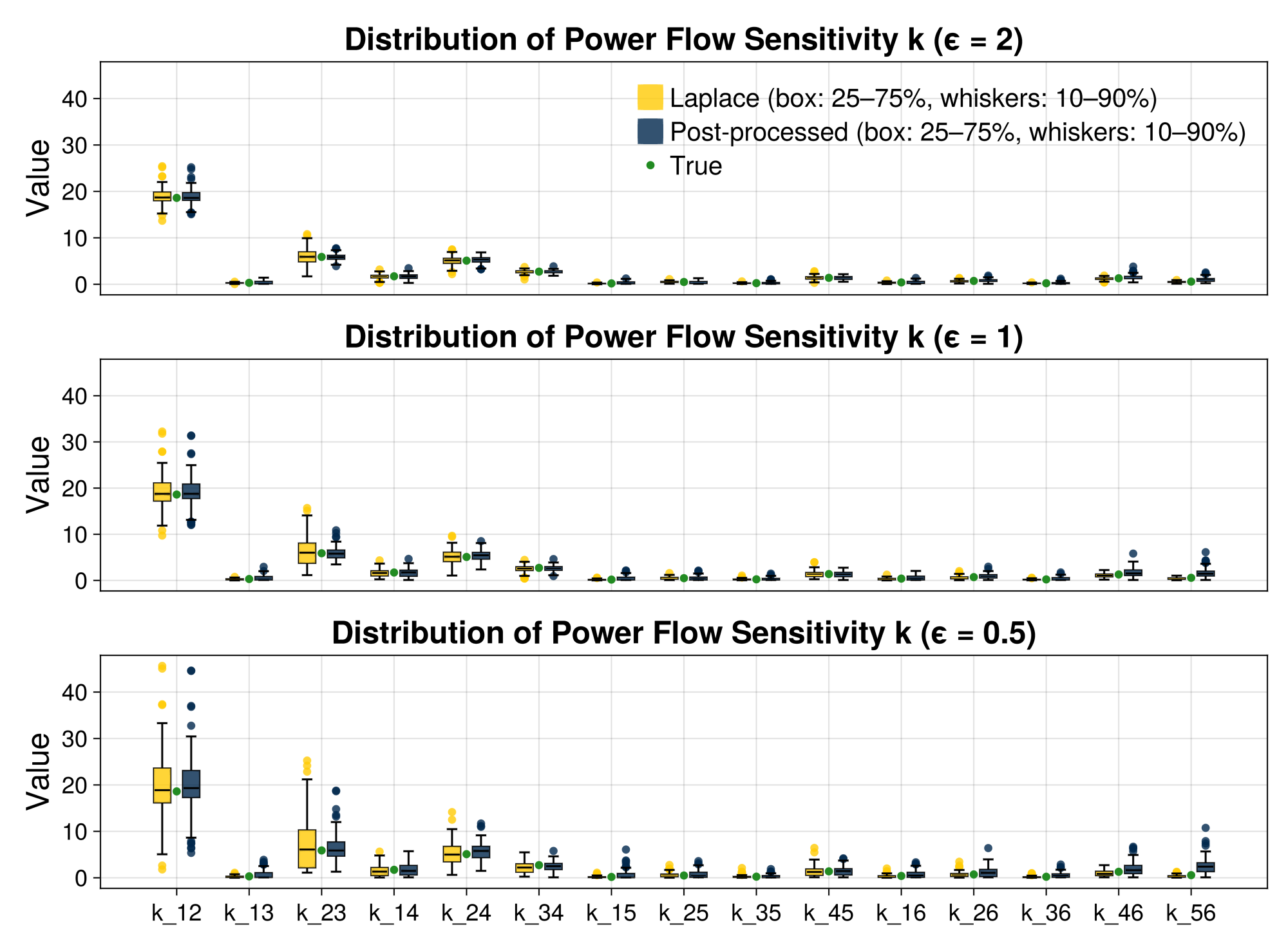}
     \caption{Distributions of power flow sensitivities $\tilde{\mb{k}}_{\text{red}}^0$ (yellow, Laplace only) and $\tilde{\mb{k}}_{\text{red}}$ (blue, post-processed) across $100$ samples, under three privacy levels $\varepsilon\in\{0.5,1,2\}$. Green dots indicate the true sensitivities $\mb{k}_{\text{red}}$.} 
     \label{fig:k_distribution}
 \end{figure}

The distribution of generator inertia is shown in Fig.~\ref{fig:M_distribution}. The interquartile ranges of both the yellow and blue boxes increase as $\varepsilon$ decreases (stronger privacy). With a weak privacy guarantee ($\varepsilon=2$), the Laplace noise magnitude is small, and the obfuscated inertia may still be close to the true inertia. With stronger privacy guarantees ($\varepsilon\in\{1,0.5\}$), the interquartile ranges and whiskers widen across all generators, making it more difficult to identify the true inertia. The inertia of generators~5 and~13 is substantially shifted by post-processing, indicating that these generators are critical for restoring the true frequency trajectories. The damping distributions in Fig.~\ref{fig:D_distribution} follow the same pattern: interquartile ranges widen as $\varepsilon$ decreases, and post-processing consistently tightens the distributions around the true damping values across all generators, indicating that damping is uniformly critical for restoring frequency trajectories.

\begin{figure}
 \centering
 \includegraphics[width=0.85\linewidth]{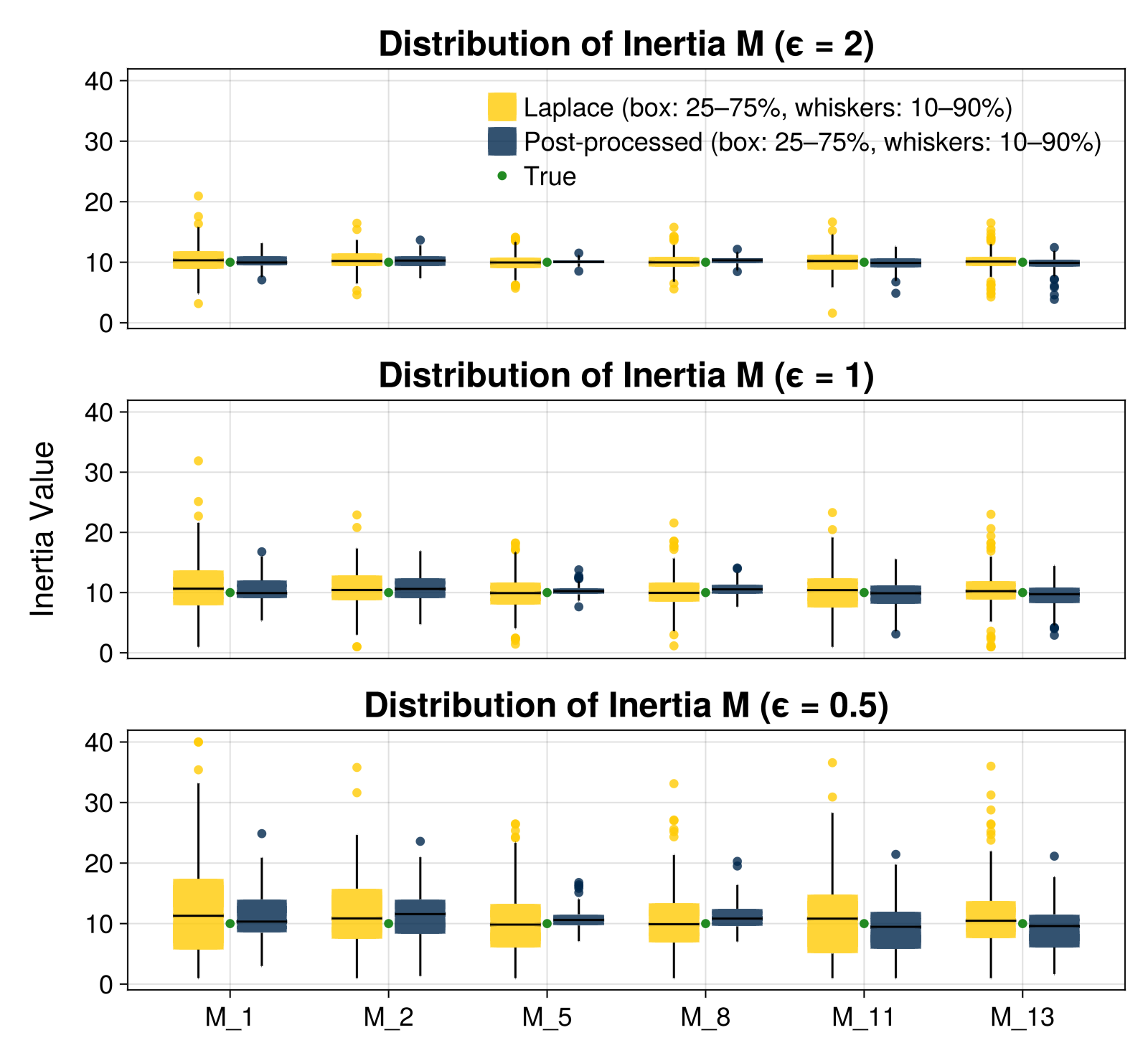}
 \caption{Distributions of generator inertia $\tilde{\mb{M}}_{\gamma}^0$ (yellow, Laplace only) and $\tilde{\mb{M}}_{\gamma}$ (blue, post-processed) across $100$ samples, under three privacy levels $\varepsilon\in\{0.5,1,2\}$. Green dots indicate the true inertia $\mb{M}_{\gamma}$.}
 \label{fig:M_distribution}
\end{figure}

\begin{figure}
 \centering
 \includegraphics[width=0.85\linewidth]{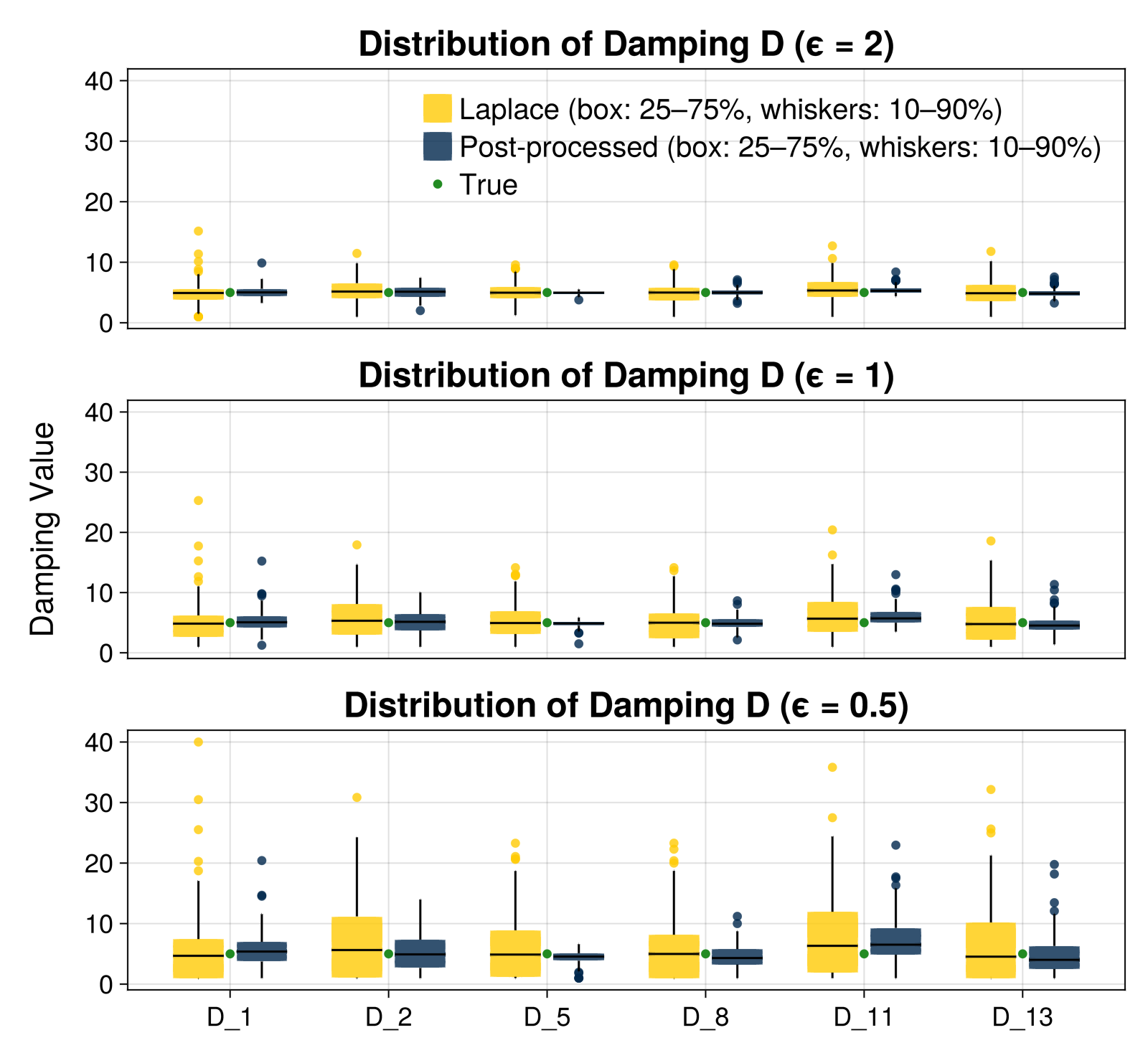}
 \caption{Distributions of generator damping $\tilde{\mb{D}}_{\gamma}^0$ (yellow, Laplace perturbation only) and $\tilde{\mb{D}}_{\gamma}$ (blue, post-processed) across $100$ samples, under three privacy levels $\varepsilon\in\{0.5,1,2\}$. Green dots indicate the true damping $\mb{D}_{\gamma}$.}
 \label{fig:D_distribution}
\end{figure}

\section{Conclusion}

We develop a principled framework for synthesizing power grid dynamic models from real-world systems with formal DP guarantees on the source grid parameters. Applied to grid frequency dynamics, the framework eliminates the risk of reverse engineering network parameters from Kron-reduced networks via calibrated perturbation of the source admittances. An ODE-constrained post-processing optimization then recovers modeling fidelity lost to DP noise while preserving the privacy guarantee, yielding synthetic models that are both privacy-preserving and useful for downstream analysis. In doing so, this work extends DP and optimization-based post-processing from steady-state grid optimization data to power grid dynamics, protecting the dynamical model itself while preserving its frequency response under load transients.

Experiments on the IEEE 30-bus system confirm the privacy-fidelity trade-off: stricter privacy requirements increase frequency trajectory deviation, yet post-processing consistently recovers fidelity across all privacy levels. The parameter distributions reveal interpretable impacts of post-processing optimization: lines and generators critical for frequency dynamics undergo larger post-processing adjustments, while less influential components retain wider variability. 

Future work will focus on testing the framework in applications to synchronous generator, PSS, AVR, and inverter models, and to linear dynamical systems more broadly.

\bibliographystyle{IEEEtran} 
\bibliography{TPWS}

\end{document}